\documentclass[conference]{IEEEtran}
\usepackage{cite}

\ifCLASSINFOpdf
  \usepackage[pdftex]{graphicx}
\else
  \usepackage[dvips]{graphicx}
\fi
\usepackage[cmex10]{amsmath}
\usepackage{amssymb,amsthm}
\usepackage{algorithm}
\usepackage{algpseudocode}

\usepackage{bm}
\usepackage{mdwmath}
\usepackage{mdwtab}
\usepackage{mathrsfs}   
\usepackage{amsfonts}
\usepackage{booktabs}

\usepackage[tight,footnotesize]{subfigure}

\usepackage{setspace}
\begin{document}
%
\title{Wideband Hybrid Precoding for Next-Generation Backhaul/Fronthaul Based on mmWave FD-MIMO}



%
\author{\IEEEauthorblockN{Yiwei Sun\IEEEauthorrefmark{1},
Zhen Gao\IEEEauthorrefmark{2},
Hua Wang\IEEEauthorrefmark{1}, and
Di Wu\IEEEauthorrefmark{3}}
\IEEEauthorblockA{\IEEEauthorrefmark{1}School of Information and Electronics, Beijing Institute of Technology, Beijing, China}
\IEEEauthorblockA{\IEEEauthorrefmark{2}Advanced Research Institute of Multidisciplinary Science, Beijing Institute of Technology, Beijing, China}
\IEEEauthorblockA{\IEEEauthorrefmark{3}China Academy of Information and Communications Technology, Beijing, China\\
\{sunyiwei, gaozhen16, wanghua\}@bit.edu.cn}}


\maketitle

\begin{abstract}
Millimeter-wave (mmWave) communication is considered as an indispensable technique for the next-generation backhaul/fronthaul network thanks to its large transmission bandwidth. Especially for heterogeneous network (HetNet), the mmWave full-dimension (FD)-MIMO is exploited to establish the backhaul/fronthaul link between phantom-cell base stations (BSs) and macro-cell BSs, where an efficient precoding is prerequisite. Against this background, this paper proposes a principle component analysis (PCA)-based hybrid precoding for wideband mmWave MIMO backhaul/fronthaul channels. We first propose an optimal hybrid precoder by exploiting principal component analysis (PCA), whereby the optimal high dimensional
frequency-selective precoder are projected
to the low-dimensional frequency-flat precoder. Moreover, the combiner is designed by leveraging the weighted PCA, where the covariance of received signal is taken into account as weight to the optimal minimum mean square error (MMSE) fully-digital combiner for further improved performance. Simulations have confirmed that the proposed scheme outperforms conventional schemes in spectral efficiency (SE) and bit-error-rate
(BER) performance.
\end{abstract}


%
\IEEEpeerreviewmaketitle

\section{Introduction}
International telecommunications union (ITU) has reached the consensus that the next-generation mobile communications will realize the goals of 1000-fold system
capacity, 100-fold energy efficiency, and 10-fold
lower latency \cite{G1,G2}. To fullfill this explosive demand of capacity, millimeter-wave (mmWave) communication with the large transmission bandwidth is usually considered to support the high-capacity backhaul/fronthaul links between phantom-cell base stations (BSs) and macro-cell BSs \cite{R3}. A typical heterogeneous network (HetNet) can be illustrated in Fig. 1 \cite{bkh}. However, mmWave communication usually suffers from the severe path loss. Traditional fully-digital precoding with massive antennas can be used to mitigate the severe path loss, but  at the cost of high power consumption and hardware cost \cite{ref_plus_5,R1,R2,R5}. To combat this issue, hybrid MIMO architecture with the much lower number of radio frequency (RF) chains than that of antennas are employed with reasonable cost and power consumption \cite{ref_plus_9,ref_plus_10,R4,R6}. Owing to the frequency flat RF precoder/combiner with constant-modulus constraint but the practical frequency selective fading mmWave channels, hybrid analog/digital precoding design can be challenging.
\begin{figure}[htb]
	\centering
	\includegraphics[width=0.7\columnwidth, keepaspectratio]{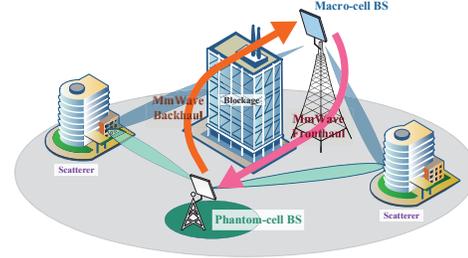}
	\caption{MmWave FD-MIMO based backhaul/fronthaul in HetNet.}
\end{figure}

Most prior mmWave hybrid precoding techniques are based on narrowband mmWave channels \cite{OMP,mao}. Specifically, a compressive sensing (CS)-based
hybrid precoding has been proposed in \cite{OMP}, where the channel sparsity is ingeniously exploited to design hybrid precoding with the aid of orthogonal matching pursuit (OMP) algorithm. To improve bit-error-rate (BER), an over-sampling codebook-based hybrid minimum sum-mean-square-error precoding is designed \cite{mao}. However in practical scenario, wideband scenario like OFDM is adopted more often \cite{lim_feedback,170825,fre_flat}. To be
specific, an insightful wideband hybrid precoder based on limited-feedback codebook has been proposed \cite{lim_feedback}. By exploiting the channel correlation information among different subcarriers, a wideband hybrid precoding is proposed not only for fully-connected structure, but also extended to the partially-connected ones \cite{170825}. Nevertheless, \cite{lim_feedback} fails to give the specific hybrid combiner at the receiver, and \cite{170825} assumes the unpractical fully-digital MIMO at the receiver. Finally, \cite{fre_flat} has theoretically shown the optimality of frequency flat precoding by proving that dominant subspaces of the frequency domain channel matrices of different subcarriers are equivalent. However, this conclusion is based on purely sparse channels with discrete angles of arrival (AoA) and angles of departure (AoD), and the practical precoder/combiner solution is not provided.

In this paper, we propose a principle component analysis (PCA) based wideband hybrid precoding scheme, which can efficiently support the mmWave FD-MIMO based backhaul/fronthaul links. To be specific, we propose the optimal PCA-based analog precoder scheme, where the low-dimensional signal space of frequency-flat RF precoder can be abstracted from the high-dimensional signal space of optimal frequency-selective precoders. Besides, the optimality of the proposed PCA-based hybrid precoding design is theoretically proven and verified. Moreover, we propose corresponding optimal weighted-PCA-based analog combiner design by extracting from fully-digital MMSE combiners. Furthermore, with the use of MMSE and equalization at digital combiner, the BER performance is improved. Simulation results show that our proposed precoding scheme have better spectral efficiency (SE) and BER performance compared to the conventional scheme.

\textsl{Notations}: Following notations are used throughout this paper. $\mathbf{A}$ is a matrix, $\mathbf{a}$ is a vector, and $a$ is a scalar. Conjugate transpose and transpose of $\mathbf{A}$ are $\mathbf{A}^H$ and $\mathbf{A}^T$, respectively. The $(i,j)$th entry of $\mathbf{A}$ is $[\mathbf{A}]_{i,j}$, and $[\mathbf{A}]_{i,:}$ ($[\mathbf{A}]_{:,j}$) denotes the $i$th row ($j$th column) of $\mathbf{A}$. Frobenius norm is denoted by $||\cdot||_F$. $|\mathbf{A}|$, $|\mathbf{a}|$, and $|a|$ are the determinant of a square matrix $\mathbf{A}$, $\ell_2$-norm of a vector $\mathbf{a}$, and modulus of a number $a$, respectively. The $i$th largest singular value of a matrix $\mathbf{A}$ is defined as $\lambda_i(\mathbf{A})$. Finally, $\mathfrak{R}\{\mathbf{A}\}$ means to keep the real part of a complex matrix $\mathbf{A}$. Additionally,
$\text{blkdiag}(\mathbf{a}_1,\cdots,\mathbf{a}_K)$ is a block diagonal matrix with $\mathbf{a}_i$ ($1\leq i\leq K$) on its diagonal blocks.

\section{System Model}
We consider an mmWave FD-MIMO system 
where both the transmitter and receiver employ the uniform planar array (UPA), and OFDM is adopted to combat the frequency-selective fading in backhaul/fronthaul channels. The transmitter is equipped with $N_t=N_t^v\times N_t^h$ antennas and $N_t^{\rm RF}\ll N_t$ chains, where $N_t^v$ and $N_t^h$ are the numbers of vertical and horizontal transmit antennas, respectively. The receiver is equipped with $N_r=N_r^v\times N_r^h$ antennas and $N_r^{\rm RF}\ll N_r$ RF chains, where $N_r^v$ and $N_r^h$ are the numbers of vertical and horizontal receive antennas, respectively. Additionally, there are $N_s$ streams transmitting in the system. We consider the downlink transmission and the received symbols at the receiver can be written as \cite{OMP}
\begin{equation}\label{rcv_sig}
\mathbf{r}[k]=(\mathbf{W}_{\rm RF}\mathbf{W}_{\rm BB}[k])^H(\mathbf{H}[k]\mathbf{F}_{\rm RF}\mathbf{F}_{\rm BB}[k]\mathbf{x}[k]+\mathbf{n}[k]),
\end{equation}
where $1\leq k\leq K$ with $K$ being the number of subcarriers, $\mathbf{F}_{\rm BB}[k]\in\mathbb{C}^{N_t^{\rm RF}\times N_s}$, $\mathbf{F}_{\rm RF}\in\mathbb{C}^{N_t\times N_t^{\rm RF}}$, $\mathbf{W}_{\rm BB}[k]\in\mathbb{C}^{N_r^{\rm RF}\times N_s}$, $\mathbf{W}_{\rm RF}\in\mathbb{C}^{N_r\times N_r^{\rm RF}}$, $\mathbf{H}[k]\in\mathbb{C}^{N_r\times N_t}$, $\mathbf{x}[k]\in\mathbb{C}^{N_s\times 1}$, and $\mathbf{n}[k]\in\mathbb{C}^{N_r\times 1}$ are the digital precoder, analog precoder, digital combiner, analog combiner, channel matrix, transmitted signal, and noise associated with the $k$th subcarrier, respectively. Noise $\mathbf{n}[k]$ satisfies $\mathbf{n}[k]\sim \mathcal{CN}(0,\sigma_n^2)$ and transmitted signal $\mathbf{x}[k]$ satisfies $\mathbb{E}[\mathbf{x}[k]\mathbf{x}^H[k]]=\frac{P}{KN_s}$, where $P$ is the average total transmit power.

The frequency-domain channel $\mathbf{H}[k]$ can be expressed as $\mathbf{H}[k]=\sum_{d=0}^{D-1}\mathbf{H}_d[d]e^{-j\frac{2\pi k}{K}d}$ \cite{lim_feedback}, where $D$ is the maximum delay spread of the discretized channels, and $\mathbf{H}_d[d]\in\mathbb{C}^{N_r\times N_t}$ is the delay-$d$ channel matrix. We consider the clustered channel model \cite{OMP}, where the channel is composed by $N_{\rm cl}$ clusters of multipaths with $N_{\rm ray}$ rays in each cluster. Thus the delay-$d$ channel matrix can be written as
\begin{equation}\label{H_d}
  \mathbf{H}_d[d]\!\!=\!\!\sum\nolimits_{i=1}^{N_{\rm cl}}\!\sum\nolimits_{l=1}^{N_{\rm ray}}\!\!p_{i,l}^d[d]\mathbf{a}_r(\phi^r_{i,l},\theta^r_{i,l})\mathbf{a}_t^H(\phi^t_{i,l},\theta^t_{i,l}),
\end{equation}
where $p_{i,l}^d[d]=\sqrt{N_tN_r/(N_{\rm cl}N_{\rm ray})}\alpha_{i,l}p(dT_s-\tau_{i,l})$ is the delay-domain channel coefficient, $\tau _{i,l}$,  $\alpha _{i,l}$, and $p(\tau)$ are the delay, the complex path gain, and the pulse shaping filter for $T_s$-spaced signaling, respectively. Thus the relationship between the frequency-domain channel coefficiency and the delay-domain channel coefficiency is $p_{i,l}[k]=\sum_{d=0}^{D-1}p_{i,l}[d]\exp(-j2\pi kd/K)$. In (\ref{H_d}), $\mathbf{a}_t(\phi^t_{i,l},\theta^t_{i,l})$ and $\mathbf{a}_r(\phi^r_{i,l},\theta^r_{i,l})$ are the steering vectors of the $l$th path in the $i$th cluster at the transmitter and receiver, respectively. In the steering vectors, $\phi^t_{i,l}$ and $\theta^t_{i,l}$ are the azimuth and elevation angles of the $l$th ray in the $i$th cluster for AoDs, and $\phi^r_{i,l}$ and $\theta^r_{i,l}$ are the azimuth and elevation angles of the $l$th ray in the $i$th cluster for AoAs. Therefore, the transmit steering vectors for the UPA at the transmitter can be expressed as $\mathbf{a}_t(\phi^t_{i,l},\theta^t_{i,l})=
[1\ \ \ \cdots \ \ \ e^{-j2\pi(m\frac{d_h}{\lambda}\sin(\theta ^t_{i,l})\cos(\phi ^t_{i,l})+n\frac{d_v}{\lambda}\sin(\phi ^t_{i,l}))} \ \ \
\cdots \\e^{-j2\pi((N_t^h-1)\frac{d_h}{\lambda}\sin(\theta ^t_{i,l})\cos(\phi ^t_{i,l})+(N_t^v-1)\frac{d_v}{\lambda}\sin(\phi ^t_{i,l}))}]^T/\sqrt{N_t}$ \cite{OMP},
where $\lambda$ is the carrier wavelength, and $d_v$ and $d_h$ are the distances between adjacent antenna elements in vertical and horizontal direction, respectively. Similarly, we can also obtain $\mathbf{a}_r(\phi ^r_{i,l},\theta^r_{i,l})$ with the same form.

Assuming the Gaussian channels, the achieved SE can be expressed as \cite{lim_feedback}
\begin{equation}
\begin{aligned}
R=\tfrac{1}{K}&\sum\nolimits_{k=1}^{K}\log_2|\mathbf{I}+
\tfrac{\rho}{N_s}\mathbf{R}_n^{-1}[k]\mathbf{W}_{\rm BB}^H[k]\mathbf{W}_{\rm RF}^H
\mathbf{H}[k]
\\&\times\mathbf{F}_{\rm RF}\mathbf{F}_{\rm BB}[k]\mathbf{F}_{\rm BB}^H[k]\mathbf{F}_{\rm RF}^H\mathbf{H}^H[k]
\mathbf{W}_{\rm RF}\mathbf{W}_{\rm BB}[k]|,
\end{aligned}
\end{equation}
where $\mathbf{R}_n[k]=\sigma_n^2\mathbf{W}_{\rm BB}^H[k]\mathbf{W}_{\rm RF}^H
\mathbf{W}_{\rm RF}\mathbf{W}_{\rm BB}[k]$ and $\rho=\frac{P}{K\sigma_n^2}$ is the signal-to-noise ratio (SNR). It is worthy to point out that our work is different from the previous work \cite{170825} with the hybrid precoder considered at the transmitter but fully-digital combiner assumed at the receiver. In this paper, we consider the hybrid MIMO architecture at both the transmitter and receiver. Our goal is to design the hybrid precoder and combiner that maximizes the SE. However, it requires to jointly optimize $R$ over variables ($\mathbf{F}_{\rm RF}$,$\{\mathbf{F}_{\rm BB}[k]\}_{k=1}^K$,$\mathbf{W}_{\rm RF}$,$\{\mathbf{W}_{\rm BB}[k]\}_{k=1}^K$) simultaneously, which is challenging. In the following sections, we will decouple the design of precoder and combiner to solve this intractable problem.

\section{Hybrid Precoder Design at Transmitter}
In this section, we discuss the design of the hybrid precoder/combiner for wideband mmWave MIMO backhaul/fronthaul channels. Our goal is to design the optimal frequency-flat RF precoder (combiner) from the optimal fully-digital precoder (combiner) for frequency-selective channels.
\subsection{Digital Precoder Design}
We first design the digital precoder by fixing the RF precoder. Specifically, we design the precoder to maximize the mutual information of the signalling as the following optimization problem
\begin{equation}
\begin{aligned}\label{opt}
\max\limits_{\mathbf{F}_{\rm RF},\mathbf{F}_{\rm BB}}&\!\!\sum\nolimits_{k=1}^{K}\log_2\!|\mathbf{I}\!+\!\tfrac{1}{\sigma_n^2}\mathbf{H}[k]\mathbf{F}_{\rm RF}\mathbf{F}_{\rm BB}[k]
\mathbf{F}_{\rm BB}^H\![k]\mathbf{F}_{\rm RF}^H\mathbf{H}^H\![k]|
\\\text{s.t. }&\mathbf{F}_{\rm RF}\in\mathcal{F}_{\rm RF},\sum\nolimits_{k=1}^K||\mathbf{F}_{\rm RF}\mathbf{F}_{\rm BB}[k]||_F^2=KN_s,
\end{aligned}
\end{equation}
where $\mathcal{F}_{\rm RF}$ is a set of feasible RF precoder satisfying constant-modulus constraint. The joint optimization of $\mathbf{F}_{\rm RF}$ and $\{\mathbf{F}_{\rm BB}[k]\}_{k=1}^K$ in (\ref{opt}) can still be difficult due to the coupling between the baseband and RF precoders \cite{lim_feedback}. Therefore, we consider $\mathbf{\widetilde{F}}_{\rm BB}[k]=(\mathbf{F}_{\rm RF}^H\mathbf{F}_{\rm RF})^{\frac{1}{2}}\mathbf{F}_{\rm BB}[k]$ to be the equivalent baseband precoder, and
the optimization problem (\ref{opt}) is equivalent to \begin{equation}
\begin{aligned}\label{opt_RF}
\max\limits_{\mathbf{F}_{\rm RF},\mathbf{\widetilde{F}}_{\rm BB}}&\sum\nolimits_{k=1}^{K}\log_2|\mathbf{I}+\tfrac{1}{\sigma_{\rm n}^2}
\mathbf{H}[k]\mathbf{F}_{\rm RF}(\mathbf{F}_{\rm RF}^H\mathbf{F}_{\rm RF})^{-\frac{1}{2}}
\\&\times\mathbf{\widetilde{F}}_{\rm BB}[k]\mathbf{\widetilde{F}}_{\rm BB}^H[k](\mathbf{F}_{\rm RF}^H\mathbf{F}_{\rm RF})^{-\frac{1}{2}}\mathbf{F}_{\rm RF}^H\mathbf{H}^H[k]|
\\\text{s.t. }&\mathbf{F}_{\rm RF}\in\mathcal{F}_{\rm RF},\sum\nolimits_{k=1}^K||\mathbf{\widetilde{F}}_{\rm BB}[k]||_F^2=KN_s.
\end{aligned}
\end{equation}

For the optimization problem (\ref{opt_RF}), we first consider the optimal solution of $\{\mathbf{\widetilde{F}}_{\rm BB}[k]\}_{k=1}^K$. Specifically, consider the singular value decomposition (SVD) of $\mathbf{H}[k]$ associated with the $k$th subcarrier as $\mathbf{H}[k]=\mathbf{U}[k]\mathbf{\Sigma}[k]\mathbf{V}^H[k]$, and the SVD of the matrix $\mathbf{\Sigma}[k]\mathbf{V}^H[k]\mathbf{F}_{\rm RF}(\mathbf{F}_{\rm RF}^H\mathbf{F}_{\rm RF})^{-1/2}=\mathbf{\widetilde{U}}[k]\mathbf{\widetilde{\Sigma}}[k]
\mathbf{\widetilde{V}}^H[k]$. Therefore, the optimal $\mathbf{\widetilde{F}}_{\rm BB}[k]=[\mathbf{\widetilde{V}}[k]]_{:,1:N_s}\mathbf{\Lambda}[k]$, and thus the optimal baseband precoder $\mathbf{F}_{\rm BB}[k]$ can be expressed as
\begin{equation}\label{F_BB}
\begin{aligned}
\mathbf{F}_{\rm BB}[k]=&(\mathbf{F}_{\rm RF}^H\mathbf{F}_{\rm RF})^{-\frac{1}{2}}\mathbf{\widetilde{F}}_{\rm BB}[k]
\\=&(\mathbf{F}_{\rm RF}^H\mathbf{F}_{\rm RF})^{-\frac{1}{2}}[\mathbf{\widetilde{V}}[k]]_{:,1:N_s}\mathbf{\Lambda}[k],
\end{aligned}
\end{equation}
where $\mathbf{\Lambda}[k]=(\mu -N_s/[\mathbf{\widetilde{\Sigma}}[k]]_{i,i}^2)^+$ ($1\leq i\leq N_s$, $1\leq k\leq K$) is a water-filling solution matrix, in which $\mu$ satisfies $\sum_{k=1}^K\sum_{i=1}^{N_s}(\mu-N_s/[\mathbf{\widetilde{\Sigma}}[k]]_{i,i}^2)^+=KN_s$.
Then the problem reduces to obtain the optimal solution of $\mathbf{F}_{\rm RF}$ to (\ref{opt_RF}).
\subsection{PCA-Based RF Precoder Design}
Previous work \cite{fre_flat} has shown that the frequency domain MIMO channel matrices $\{\mathbf{H}[k]\}_{k=1}^K$ have the same column space and row space. Meanwhile, the frequency-flat RF precoder $\mathbf{F}_{\rm RF}$ remains unchanged for all subcarriers. So the RF precoder can be regarded as a representation of such a column space. This observation motivates us to design the RF precoder by leveraging the PCA \cite{mach_lern}. Specifically, we first define the optimal digital precoder $\mathbf{F}_{\rm opt}[k]=[\mathbf{V}[k]]_{:,1:N_s}$ for $1\leq k\leq K$. Moreover, we regard the matrix $\mathbf{F}=\begin{bmatrix}\mathbf{F}_{\rm opt}[1]\ \mathbf{F}_{\rm opt}[2]\ \cdots\ \mathbf{F}_{\rm opt}[K]\end{bmatrix}$
consisting of the optimal precoders of all subcarriers as the data set
in the PCA problem. Additionally, to achieve the stable solution with low complexity for PCA,
SVD is applied to the data set matrix $\mathbf{F}$ \cite{mach_lern}. This process is detailed in Proposition 1, where its optimality is also verified as follows.
\newtheorem{prop}{\textbf{Proposition}}
\begin{prop}
    Considering $\mathbf{F}=\begin{bmatrix}\mathbf{F}_{\rm opt}[1]\ \mathbf{F}_{\rm opt}[2]\ \cdots\ \mathbf{F}_{\rm opt}[K]\end{bmatrix}$ and its SVD $\mathbf{F}=\mathbf{U}_F\mathbf{\Sigma}_F\mathbf{V}_F^H$, the solution to (\ref{opt_RF}) can be expressed as $\mathbf{F}_{\rm RF}=[\mathbf{U}_F]_{:,1:N_t^{\rm RF}}\mathbf{R}_t$, where $\mathbf{R}_t\in\mathbb{C}^{N_t^{\rm RF}
	\times N_t^{\rm RF}}$ is an arbitrary full rank matrix.
\end{prop}
\begin{proof}
	Following the similar steps of the equations (12)-(14) in \cite{OMP} and defining $[\mathbf{\Sigma}[k]]_{1:N_s,1:N_s}=\mathbf{\Sigma}_1[k]$, the objective function of the optimization problem in (\ref{opt}) can be approximately written as
\begin{equation}
\begin{aligned}
&\sum\nolimits_{k=1}^{K}\log_2|\mathbf{I}+\tfrac{1}{\sigma_n^2}\mathbf{H}[k]\mathbf{F}_{\rm RF}\mathbf{F}_{\rm BB}[k]
\mathbf{F}_{\rm BB}^H[k]\mathbf{F}_{\rm RF}^H\mathbf{H}^H[k]|
\\\approx&\!\sum\nolimits_{k=1}^{K}\!(\log_2\!|\mathbf{I}_{N_s}\!\!+\!\!\tfrac{1}{\sigma_n^2}
\!\mathbf{\Sigma}_1^2[k]|
\!-\!(\!N_s\!\!-\!\!||\mathbf{F}_{\rm opt}^H\![k]\mathbf{F}_{\rm RF}\mathbf{F}_{\rm BB}[k]||_F^2)\!).
\end{aligned}
\end{equation}
Therefore, the optimization problem (\ref{opt}) is equivalent to the following optimization problem
\begin{equation} \label{opt_final}
\begin{aligned}
\max\limits_{\mathbf{F}_{\rm RF},\mathbf{F}_{\rm BB}}&\sum\nolimits_{k=1}^{K}||\mathbf{F}_{\rm opt}^H[k]\mathbf{F}_{\rm RF}\mathbf{F}_{\rm BB}[k]||_F^2
\\\text{s.t. }&\mathbf{F}_{\rm RF}\in\mathcal{F}_{\rm RF},\sum\nolimits_{k=1}^K||\mathbf{F}_{\rm RF}\mathbf{F}_{\rm BB}[k]||_F^2=KN_s,
\end{aligned}
\end{equation}
where $\mathcal{F}_{\rm RF}$ is a set of feasible RF precoder satisfying constant-modulus constraint.
The objective function in (\ref{opt_final}) can be written as
\begin{equation}
\begin{aligned}
\sum\nolimits_{k=1}^{K}&\!\!\!\!\!\!
||\mathbf{F}_{\rm opt}^H\![k]\mathbf{F}_{\rm RF}\mathbf{F}_{\rm BB}[k]||_F^2
\!=\!\!\sum\nolimits_{k=1}^{K}\!\!\!\!\!\!\!\!\text{Tr}(\mathbf{F}_{\rm opt}^H\![k]\mathbf{F}_{\rm RF}\!(\mathbf{F}_{\rm RF}^H\mathbf{F}_{\rm RF}\!)\!^{-\!\frac{1}{2}}
\\&\times(\mathbf{\widetilde{F}}_{\rm BB}[k]\mathbf{\widetilde{F}}_{\rm BB}^H[k])
(\mathbf{F}_{\rm RF}^H\mathbf{F}_{\rm RF})^{-\frac{1}{2}}\mathbf{F}_{\rm RF}^H\mathbf{F}_{\rm opt}[k]).
\end{aligned}
\end{equation}
According to previous work \cite{uni_cons}, unitary constraints offer a close performance to the total power constraint while providing a relatively simple form of solution. To simplify the problem, we consider condition under unitary power constraints instead. Therefore, water-filling power allocation coefficients can be ignored. In detail, the equivalent baseband precoder $\mathbf{\widetilde{F}}_{\rm BB}[k]=[\mathbf{\widetilde{V}}[k]]_{:,1:N_s}$, which means that $\mathbf{\widetilde{F}}_{\rm BB}[k]$ is a unitary or simi-unitary matrix depending on the relationship between $N_s$ and $N_t^{\rm RF}$. When $N_s=N_t^{\rm RF}$, $\mathbf{\widetilde{F}}_{\rm BB}[k]\mathbf{\widetilde{F}}_{\rm BB}^H[k]$ is $\mathbf{I}_{N_s}$. When $N_s<N_t^{\rm RF}$, denoting the SVD of $\mathbf{\widetilde{F}}_{\rm BB}[k]=\mathbf{U}_{\rm BB}[k]\begin{bmatrix}
\mathbf{I_{N_s}} \ \mathbf{0} \end{bmatrix}^T\mathbf{V}_{\rm BB}^H[k]$, thus $\mathbf{\widetilde{F}}_{\rm BB}[k]\mathbf{\widetilde{F}}_{\rm BB}^H[k]=\mathbf{U}_{\rm BB}[k]\text{blkdiag}(\mathbf{I}_{N_s},\mathbf{0}_{N_t^{\rm RF}-N_s})\mathbf{U}_{\rm BB}^H[k]$. Therefore, the solution to the condition when $N_s=N_t^{\rm RF}$ will also suffice the condition when $N_s<N_t^{\rm RF}$. Therefore, the (\ref{opt}) goes down to $\sum_{k=1}^{K}||\mathbf{F}_{\rm opt}^H[k]\mathbf{F}_{\rm RF}(\mathbf{F}_{\rm RF}^H\mathbf{F}_{\rm RF})^{-\frac{1}{2}}||_F^2$.
Considering SVD of $\mathbf{F}_{\rm RF}=\mathbf{U}_{\rm RF}\mathbf{\Sigma}_{\rm RF}\mathbf{V}_{\rm RF}^H$, previous function can be written as
\begin{equation}\label{s_eq_RF}
\begin{aligned}
&\sum_{k=1}^{K}||\mathbf{F}_{\rm opt}^H[k]\mathbf{F}_{\rm RF}(\mathbf{F}_{\rm RF}^H\mathbf{F}_{\rm RF})^{-\frac{1}{2}}||_F^2=\sum_{k=1}^{K}
||\mathbf{F}_{\rm opt}^H[k]\mathbf{U}_{\rm RF}||_F^2
\\=&\text{Tr}(\sum_{k=1}^{K}\mathbf{U}_{\rm RF}^H\mathbf{F}_{\rm opt}[k]\mathbf{F}_{\rm opt}^H[k]\mathbf{U}_{\rm RF})
\\=&\text{Tr}(\begin{bmatrix}
\mathbf{U}_{\rm RF}^H\mathbf{F}_{\rm opt}[1] & \cdots & \mathbf{U}_{\rm RF}^H\mathbf{F}_{\rm opt}[K]
\end{bmatrix}
\begin{bmatrix}
\mathbf{F}_{\rm opt}^H[1]\mathbf{U}_{\rm RF} \\
\vdots \\
\mathbf{F}_{\rm opt}^H[K]\mathbf{U}_{\rm RF}
\end{bmatrix})
\\=&\text{Tr}(\mathbf{U}_{\rm RF}^H\mathbf{F}\mathbf{F}^H\mathbf{U}_{\rm RF})=\text{Tr}(\mathbf{U}_{\rm RF}^H\mathbf{U}_F\mathbf{\Sigma}_F^2\mathbf{U}_F^H\mathbf{U}_{\rm RF}).
\end{aligned}
\end{equation}
Since both $\mathbf{U}_{\rm RF}$ and $\mathbf{U}_F$ are unitary matrix, (\ref{s_eq_RF}) reaches the maximum only when $\mathbf{U}_{\rm RF}=\mathbf{U}_F$. Moreover, the rank of $\mathbf{F}_{\rm RF}$ is $N_t^{\rm RF}$. So, it is safe to say that $\mathbf{U}_{\rm RF}=[\mathbf{U}_F]_{:,1:N_t^{\rm RF}}\mathbf{U}_R$. Hence the optimal RF precoder can be expressed as
\begin{equation}\label{F_RF}
\mathbf{F}_{\rm RF}=[\mathbf{U}_F]_{:,1:N_t^{\rm RF}}\mathbf{U}_R\mathbf{\Sigma}_{\rm RF}\mathbf{V}_{\rm RF}^H=\mathbf{U}_F\mathbf{R}_t,
\end{equation}
where $\mathbf{R}_t=\mathbf{U}_R\mathbf{\Sigma}_{\rm RF}\mathbf{V}_{\rm RF}^H\in\mathbb{C}^{N_t^{\rm RF}
\times N_t^{\rm RF}}$ is an arbitrary full rank matrix.
\end{proof}

According to Proposition 1, we can obtain the principal
components constituting the optimal RF precoder using PCA.
Moreover, we will design the full-rank matrix $\mathbf{R}_t$ to
meet the requirement of constant-modulus constraint for RF
precoder. Specifically, by taking the constant-modulus constraint of RF precoder into account, we can design the RF precoder by solving
\begin{equation}\label{angle_sel}
\mathbf{F}_{\rm RF}=\arg\min_{|[\mathbf{X}]_{i,j}|=1/\sqrt{N_t}}||\mathbf{X}-
[\mathbf{U}_F]_{:,1:N_t^{\rm RF}}||^2_F.
\end{equation}
With the constant-modulus constrains, the set of possible $\mathbf{F}_{\rm RF}$ is actually a hypersphere in the space of $\mathbb{C}^{N_t\times N_t^{\rm RF}}$, and $[\mathbf{U}_f]_{:,1:N_t^{\rm RF}}$ is a known point in the space of $\mathbb{C}^{N_t\times N_t^{\rm RF}}$. Therefore, the optimization problem in (\ref{angle_sel}) is actually a distance minimization problem. Naturally, the solution is the point on this hypersphere sharing same direction of the know point. In other words, the solution is given by $[\mathbf{F}_{\rm RF}]_{i,j}=1/\sqrt{N_t}e^{j\angle([\mathbf{U}_f]_{i,j})}$, and $\angle(\alpha)$ denotes the phase of a complex number $\alpha$.The specific RF precoder design is summarized in Algorithm \ref{alg:PCA_AS}.
\begin{algorithm}[htb]
	\caption{ PCA-based RF Precoder Design.}
	\label{alg:PCA_AS}
	\begin{algorithmic}[1]
		\renewcommand{\algorithmicrequire}{\textbf{Input:}}
		\renewcommand\algorithmicensure {\textbf{Output:}}
		\Require
		Optimal precoder $\mathbf{F}_{\rm opt}$,
		number of RF chains $N_t^{\rm RF}$，
		number of antennas $N_t$.
		\Ensure
		RF precoder $\mathbf{F}_{\rm RF}$.
		\State $\mathbf{F}=\begin{bmatrix}\mathbf{F}_{\rm opt}[1]\ \mathbf{F}_{\rm opt}[2]\ \cdots\ \mathbf{F}_{\rm opt}[K]\end{bmatrix}$
		\State Apply SVD to $\mathbf{F}$, i.e., $\mathbf{F}=\mathbf{U}_F\mathbf{\Sigma}_F\mathbf{V}_F^H$, where $\mathbf{U}_F$ corresponds to the principal components
		\State $[\mathbf{F}_{\rm RF}]_{i,j}=\tfrac{1}{\sqrt{N_t}}e^{j\angle([\mathbf{U}_F]_{i,j})}$
	\end{algorithmic}
\end{algorithm}

When the quantization of phase shifters is considered, we assume the quantization bits are $Q$. Therefore, the phase shifters can only be chosen from the following quantized phase set $\mathcal{Q}=\{0,\frac{2\pi}{2^Q},\cdots,\frac{2\pi(2^Q-1)}{2^Q}\}$. Specifically, after obtaining the RF precoder $\mathbf{F}_{\rm RF}$, the quantization process can be realized by searching for the minimum Euclidean distance between $\angle([\mathbf{F}_{\rm RF}]_{i,j})$ and quantized phase from $\mathcal{Q}$.
\section{Hybrid Combiner Design at Receiver}
In this section, we assume that $\mathbf{F}_{\rm RF}$ and $\{\mathbf{F}_{\rm BB}[k]\}_{k=1}^K$ are fixed and seek to design the hybrid combiner to minimize the mean-square-error (MSE) between the received signal and the transmitted signal \cite{OMP}. Specifically, the optimal fully-digital minimum mean square error (MMSE) combiner can be expressed as
\begin{equation}
\begin{aligned}\label{W_opt}
\mathbf{W}_{\rm opt}^H[k]\!=\!&\mathbf{W}_{\rm MMSE}^H[k]
\!=\!\tfrac{\sqrt{\rho}}{N_s}\mathbf{F}_{\rm BB}^H[k]\mathbf{F}_{\rm RF}^H\mathbf{H}^H[k]
(\!\tfrac{\rho}{N_s}\mathbf{H}[k]\mathbf{F}_{\rm RF}
\\&\times\mathbf{F}_{\rm BB}[k]\mathbf{F}_{\rm BB}^H[k]\mathbf{F}_{\rm RF}^H\mathbf{H}^H[k]+\sigma_{\rm n}^2\mathbf{I}_{N_r})^{-1},
\end{aligned}
\end{equation}
for $1\leq k\leq K$. Denoting the signal at the receiving antenna as $\mathbf{y}\in\mathbb{C}^{N_r\times 1}$, the combiner design MSE problem is
\begin{equation}\label{MSE}
\begin{aligned}
	\min\limits_{\mathbf{W}_{\rm RF},\mathbf{W}_{\rm BB}}&\sum\nolimits_{k=1}^{K}\mathbb{E}[||\mathbf{x}[k]-\mathbf{W}_{\rm BB}^H[k]\mathbf{W}_{\rm RF}^H\mathbf{y}[k]||_2^2]
	\\\text{s.t. }&\mathbf{W}_{\rm RF}\in\mathcal{W}_{\rm RF},
\end{aligned}
\end{equation}
where $\mathcal{W}_{\rm RF}$ is a set of feasible RF precoder satisfying constant-modulus constraint. Note that if the constant-modulus
constraint in (\ref{MSE}) is removed, the solution to (\ref{MSE}) is the optimal fully-digital
MMSE combiner in (\ref{W_opt}). On the other hand, we observe that the objective
function in (\ref{MSE}) can be further expressed as
\begin{equation}\label{inner_MSE}
\begin{aligned}
	&\sum\nolimits_{k=1}^{K}\mathbb{E}[||\mathbf{x}[k]-\mathbf{W}_{\rm BB}^H[k]\mathbf{W}_{\rm RF}^H\mathbf{y}[k]||_2^2]
	\\=&\!\sum_{k=1}^{K}\!\text{Tr}(\mathbb{E}[\mathbf{x}[k]\mathbf{x}\!^H\![k]])
	\!-\!2\!\sum_{k=1}^{K}\!\mathcal{R}\{\text{Tr}(\mathbb{E}[\mathbf{x}[k]\mathbf{y}\!^H\![k]]\mathbf{W}\!_{\rm RF}\!\mathbf{W}\!_{\rm BB}[k]\!)\!\}
	\\&+\text{Tr}(\mathbf{W}_{\rm BB}^H[k]\mathbf{W}_{\rm RF}^H\mathbb{E}[\mathbf{y}[k]\mathbf{y}^H[k]]\mathbf{W}_{\rm RF}\mathbf{W}_{\rm BB}[k]).
\end{aligned}
\end{equation}
Since the optimization variables in (\ref{MSE}) are $\mathbf{W}_{\rm RF}$ and $\{\mathbf{W}_{\rm BB}[k]\}_{k=1}^K$, any term independent
with $\mathbf{W}_{\rm RF}$ and $\{\mathbf{W}_{\rm BB}[k]\}_{k=1}^K$ will not influence the outcome. Thus we add the independent term $\sum_{k=1}^{K}\text{Tr}(\mathbf{W}_{\rm opt}^H[k]\mathbb{E}[\mathbf{y}[k]\mathbf{y}^H[k]]\\\times\mathbf{W}_{\rm opt}[k])-\sum_{k=1}^{K}\text{Tr}(\mathbb{E}[\mathbf{x}[k]\mathbf{x}^H[k]])$ to the objective function (\ref{inner_MSE}). So the objective
function in (\ref{MSE}) can be rewritten as
\begin{equation}\label{problem_W}
\sum\nolimits_{k=1}^{K}||\mathbb{E}[\mathbf{y}[k]\mathbf{y}^H[k]]^{\frac{1}{2}}(\mathbf{W}_{\rm opt}[k]\!\!-\!\!\mathbf{W}_{\rm RF}\mathbf{W}_{\rm BB}[k])||_F^2
\end{equation}
where $\mathbb{E}[\mathbf{y}[k]\mathbf{y}^H[k]]=
(\rho/N_s)\mathbf{H}[k]\mathbf{F}_{\rm RF}\mathbf{F}_{\rm BB}[k]\mathbf{F}_{\rm BB}^H[k]
\cdot\mathbf{F}_{\rm RF}^H\mathbf{H}^H[k]+\sigma_{\rm n}^2\mathbf{I}_{N_r}$ and $\mathcal{W}_{\rm RF}$ is the constant-modulus constraint. Similar to the precoder design, we derive the structure of the optimal RF combiners that solves (\ref{problem_W}).
\begin{prop}
	Considering $\mathbf{W}=\begin{bmatrix}
    \mathbb{E}[\mathbf{y}[1]\mathbf{y}[1]^H]^{1/2}\mathbf{W}_{\rm opt}[1] \ \cdots \ \mathbb{E}[\mathbf{y}[K]\mathbf{y}[K]^H]^{1/2}\mathbf{W}_{\rm opt}[K]
	\end{bmatrix}$ and its SVD $\mathbf{W}=\mathbf{U}_W\mathbf{\Sigma}_W\mathbf{V}_W^H$, the solution to (\ref{problem_W}) can be written as $\mathbf{W}_{\rm RF}=[\mathbf{U}_W]_{:,1:N_r^{\rm RF}}\mathbf{R}_r$, where $\mathbf{R}_r\in\mathbb{C}^{N_r^{\rm RF}
	\times N_r^{\rm RF}}$ is an arbitrary full rank matrix.
\end{prop}
\begin{proof}
	Consider least square (LS) estimation for baseband combiner $\mathbf{W}_{\rm BB}[k]=(\mathbf{W}_{\rm RF}^H\mathbb{E}[\mathbf{y}[k]\mathbf{y}^H[k]]\mathbf{W}_{\rm RF})^{-1}\mathbf{W}_{\rm RF}^H\mathbb{E}[\mathbf{y}[k]\mathbf{y}^H[k]]\mathbf{W}_{\rm opt}[k]$ ($k=1,\cdots,K$). Substituting LS estimation for baseband combiner in the objective function of (\ref{problem_W}), we have
\begin{equation}
\begin{aligned}
&\sum_{k=1}^{K}||\mathbb{E}[\mathbf{y}[k]\mathbf{y}^H[k]]^{\frac{1}{2}}(\mathbf{W}_{\rm opt}[k]-\mathbf{W}_{\rm RF}\mathbf{W}_{\rm BB}[k])||_F^2
\\=&\sum_{k=1}^{K}||\mathbb{E}[\mathbf{y}[k]\mathbf{y}^H[k]]^{\frac{1}{2}}\mathbf{W}_{\rm opt}[k]-\mathbb{E}[\mathbf{y}[k]\mathbf{y}^H[k]]^{\frac{1}{2}}\mathbf{W}_{\rm RF}(\mathbf{W}_{\rm RF}^H
\\&\times\mathbb{E}[\mathbf{y}[k]\mathbf{y}^H[k]]\mathbf{W}_{\rm RF})^{-1}\mathbf{W}_{\rm RF}^H\mathbb{E}[\mathbf{y}[k]\mathbf{y}^H[k]]\mathbf{W}_{\rm opt}[k])||_F^2.
\end{aligned}
\end{equation}
Let $\mathbf{A}[k]=\mathbb{E}[\mathbf{y}[k]\mathbf{y}^H[k]]^{\frac{1}{2}}\mathbf{W}_{\rm opt}[k]$ and $\mathbf{B}[k]=\mathbb{E}[\mathbf{y}[k]\mathbf{y}^H[k]]^{\frac{1}{2}}\mathbf{W}_{\rm RF}$. The above equation can be further simplified as
\begin{equation}
\sum_{k=1}^{K}(\text{Tr}(\mathbf{A}\!^H\![k]\mathbf{A}[k]\!)
-\!\text{Tr}(\mathbf{A}\!^H\![k]\mathbf{B}[k](\mathbf{B}\!^H\![k]
\mathbf{B}[k])\!^{-1}\mathbf{B}\!^H\![k]\mathbf{A}[k])\!).
\end{equation}
So the minimization problem can be transformed into following maximization problem
\begin{equation}\label{max_W}
\max_{\mathbf{W}_{\rm RF},\mathbf{W}_{\rm BB}[k]}\sum_{k=1}^{K}\text{Tr}(\mathbf{A}^H[k]\mathbf{B}[k](\mathbf{B}^H[k]
\mathbf{B}[k])^{-1}\mathbf{B}^H[k]\mathbf{A}[k]).
\end{equation}
Assuming the SVD of $\mathbf{B}[k]=\mathbf{U}_B[k]\mathbf{\Sigma}_B[k]\mathbf{V}_B^H[k]$, the objective function of the maximization problem (\ref{max_W}) can be written as $\sum_{k=1}^{K}||\mathbf{U}_B^H[k]\mathbf{A}[k]||_F^2$.
Since $\mathbf{B}[k]=\mathbb{E}[\mathbf{y}[k]\mathbf{y}^H[k]]^{\frac{1}{2}}\mathbf{W}_{\rm RF}$, we can safe to say that there is a matrix $\mathbf{R}_B[k]\in\mathbb{C}^{N_r^{\rm RF}
\times N_r^{\rm RF}}$ satisfying $\mathbf{U}_B[k]=\mathbf{R}_B[k]\mathbf{U}_W$. So the problem can be transformed into $\sum_{k=1}^{K}||\mathbf{U}_W^H\mathbf{R}_B^H[k]\mathbf{A}[k]||_F^2$.
Note that function above is similar to the objective function of maximization problem (\ref{opt_final}), thus $\mathbf{W}_{\rm RF}=[\mathbf{U}_V]_{:,1:N_r^{\rm RF}}\mathbf{R}_v$, where
\begin{equation}
\begin{aligned}
\mathbf{V}=&[\mathbf{R}_B^H[1]\mathbf{A}[1]\ \cdots \ \mathbf{R}_B^H[K]\mathbf{A}[K]]
\\=&\text{blkdiag}(\mathbf{R}_B^H[1],\cdots,\mathbf{R}_B^H[K])\mathbf{W}.
\end{aligned}
\end{equation}
Therefore, matrix $\mathbf{U}_R$ satisfies $[\mathbf{U}_V]_{:,1:N_r^{\rm RF}}=[\mathbf{U}_W]_{:,1:N_r^{\rm RF}}\mathbf{U}_R$. So the solution to the problem (\ref{max_W}) is
\begin{equation}
\begin{aligned}
&\mathbf{W}_{\rm RF}=[\mathbf{U}_V]_{:,1:N_r^{\rm RF}}\mathbf{R}_v
\\=&[\mathbf{U}_W]_{:,1:N_r^{\rm RF}}\mathbf{U}_R\mathbf{R}_v=[\mathbf{U}_W]_{:,1:N_r^{\rm RF}}\mathbf{R}_r,
\end{aligned}
\end{equation}
where $\mathbf{R}_r\in\mathbb{C}^{N_r^{\rm RF}
\times N_r^{\rm RF}}$ is an arbitrary full rank matrix.
\end{proof}

The design of RF combiner can be extended from that of the RF precoder, where the weight $\mathbb{E}[\mathbf{y}[k]\mathbf{y}^H[k]]$ should be considered according to MMSE criterion. The RF combiner design is provided in Algorithm \ref{alg:w_PCA_AS}.
\begin{algorithm}[htb]
	\caption{ Weighted PCA-based RF Combiner Design.}
	\label{alg:w_PCA_AS}
	\begin{algorithmic}[1]
		\renewcommand{\algorithmicrequire}{\textbf{Input:}}
		\renewcommand\algorithmicensure {\textbf{Output:}}
		\Require
		Optimal combiners $\{\mathbf{W}_{\rm opt}[k]\}_{k=1}^K$,
        covariance matrices of the received signals $\{\mathbb{E}[\mathbf{y}[k]\mathbf{y}^H[k]]^{1/2}\}_{k=1}^K$,
		the number of RF chains $N_r^{\rm RF}$,
		and the number of antennas $N_r$.
		\Ensure
		RF combiner $\mathbf{W}_{\rm RF}$.
        \State $\mathbf{W}_w[k]=\mathbb{E}[\mathbf{y}[k]\mathbf{y}^H[k]]^{1/2}\mathbf{W}_{\rm opt}[k]$, for $k=1,\cdots,K$
        \State $\mathbf{W}=\begin{bmatrix}\mathbf{W}_w[1]\ \mathbf{W}_w[2]\ \cdots\ \mathbf{W}_w[K]\end{bmatrix}$
		\State Apply SVD to $\mathbf{W}$, i.e., $\mathbf{W}=\mathbf{U}_W\mathbf{\Sigma}_W\mathbf{V}_W^H$, where $\mathbf{U}_W$ corresponds to the principal components
		\State $[\mathbf{W}_{\rm RF}]_{i,j}=\frac{1}{\sqrt{N_r}}e^{j\angle([\mathbf{U}_W]_{i,j})}$
	\end{algorithmic}
\end{algorithm}
When consider quantization of phase shifters, the entries of $\angle([\mathbf{W}_{\rm RF}]_{i,j})$ are substituted by the phase in $Q$ bits quantization phase set $\mathcal{Q}$ with minimum Euclidean distance.

Furthermore, the design of baseband combiners $\{\mathbf{W}_{\rm BB}[k]\}_{k=1}^K$ are different from the design of baseband precoders $\{\mathbf{F}_{\rm BB}[k]\}_{k=1}^K$, since the power constraint is removed for the receive hybrid combiner. Specifically, $\{\mathbf{F}_{\rm BB}[k]\}_{k=1}^K$ is designed according to water-filling algorithm given $\mathbf{F}_{\rm RF}$, while $\{\mathbf{W}_{\rm BB}[k]\}_{k=1}^K$ is designed by using weighted least squares (LS) according to fully-digital MMSE combiners $\{\mathbf{W}_{\rm opt}[k]\}_{k=1}^K$ and frequency-flat RF combiner $\mathbf{W}_{\rm RF}$. The detailed design of $\{\mathbf{W}_{\rm BB}[k]\}_{k=1}^K$ can be summarized in Algorithm \ref{alg:comb}.
\begin{algorithm}[htb]
	\caption{ Baseband Combiner Design.}
	\label{alg:comb}
	\begin{algorithmic}[1]
		\renewcommand{\algorithmicrequire}{\textbf{Input:}}
		\renewcommand\algorithmicensure {\textbf{Output:}}
		\Require
		Optimal combiners $\{\mathbf{W}_{\rm opt}[k]\}_{k=1}^K$,
		RF combiner $\mathbf{W}_{\rm RF}$,
		RF precoder $\mathbf{F}_{\rm RF}$,
        baseband precoders $\{\mathbf{F}_{\rm BB}[k]\}_{k=1}^K$,
		channel matrices $\{\mathbf{H}[k]\}_{k=1}^K$,
		and expectation of the received signal $\{\mathbb{E}[\mathbf{y}[k]\mathbf{y}^H[k]]\}_{k=1}^K$.
		\Ensure
		Baseband combiners $\{\mathbf{W}_{\rm BB}[k]\}_{k=1}^K$.
		\For{$k=1:K$}
		\State $\mathbf{A}=(\mathbf{W}_{\rm RF}^H\mathbb{E}[\mathbf{y}[k]\mathbf{y}^H[k]]\mathbf{W}_{\rm RF})^{-1}$
        \State $\mathbf{W}_{\rm BB}[k]=\mathbf{A}\mathbf{W}_{\rm RF}^H\mathbb{E}[\mathbf{y}[k]\mathbf{y}^H[k]]\mathbf{W}_{\rm opt}[k]$
		\State $\mathbf{\Lambda}_{\rm eq}=\text{diag}{\{\mathbf{W}_{\rm BB}^H[k] \mathbf{W}_{\rm RF}^H\mathbf{H}[k]\mathbf{F}_{\rm RF}\mathbf{F}_{\rm BB}[k]\}}^{-1}$
		\State $\mathbf{W}_{\rm BB}[k]=\mathbf{W}_{\rm BB}[k]\mathbf{\Lambda}_{\rm eq}$
		\EndFor
	\end{algorithmic}
\end{algorithm}

\section{Simulations}
In this section, we will investigate the SE and BER performance for the hybrid precoder/combiner design for the backhaul/fronthaul channel. For the backhaul/fronthaul channel model, we adopt Dirac delta function as the pulse shaping filter and a cyclic prefix with the length of $D=64$. The number of subcarriers is $K=512$. We consider that the path delay is uniformly distributed in $[0,DT_s]$ ($T_s=1/B$ is the symbol period). The number of the clusters is $N_{\rm cl}=8$, and azimuth/elevation AoAs and AoDs follow the uniform distribution $\mathcal{U}[-\pi/2, \pi/2]$ with angle spread of $7.5^{\circ}$. Within each cluster, there are $N_{\rm ray}=10$ rays. As for the antennas, we consider both the transmitter and receiver adopt $8\times8$ UPA, and the distance between each adjacent antennas is half wavelength. 
Moreover, we consider the number of RF chains at transmitter and receiver are $N_t^{\rm RF}=N_r^{\rm RF}=4$ and the data stream is $N_s=3$ unless otherwise stated.

Throughout this part, following baselines will be considered
for performance benchmarks: \textbf{Optimal fully-digital} scheme considers the fully-digital MIMO system, where the SVD-based precoder/combiner is adopted as the performance upper bound. \textbf{Simultaneous OMP (SOMP)} scheme is an extension version of the narrow-band OMP-based spatially sparse precoding in \cite{OMP}. In broadband, SOMP-based hybrid precoding scheme can simultaneously design the RF precoder/combiner for all subcarriers. \textbf{Discrete Fourier transform (DFT) codebook} scheme designs the RF precoder/combiner from the DFT codebook instead of steering vectors codebook in SOMP scheme \cite{DFT_cb}.

\begin{figure}[htb]
    \centering
    \includegraphics[width=0.8\columnwidth, keepaspectratio]{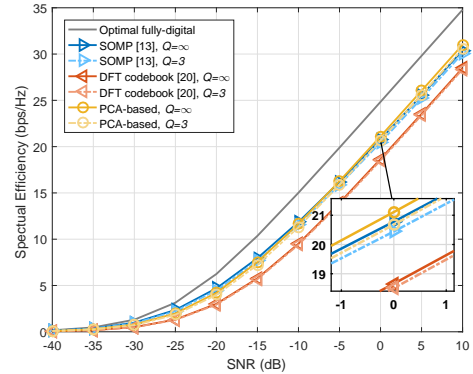}
	\caption{SE performance comparison of different hybrid precoding schemes, where both transmitter and receiver employ $8\times8$ UPA, $N_t^{\rm RF}=N_r^{\rm RF}=4$, and $K=512$.}\label{fig_SE}
\end{figure}
In Fig. \ref{fig_SE}, we evaluate SE performance of the proposed PCA-based hybrid precoding scheme, SOMP-based hybrid precoding scheme, and DFT codebook-based hybrid precoding scheme, where both hybrid precoder and combiner are jointly investigated, and $Q=\infty$ and $Q=3$ are considered, respectively. In Fig. \ref{fig_SE}, the proposed PCA-based hybrid precoding scheme outperforms the SOMP-based hybrid precoding scheme and DFT codebook-based hybrid precoding scheme. The DFT codebook-based hybrid precoding scheme works poorly due to the quantization loss of the DFT codebook with limited size. Finally, it can also be observed from Fig. \ref{fig_SE} that the influence of quantization of phase shifters can be negligible for our scheme.

\begin{figure}[htb]
    \centering
    \includegraphics[width=0.8\columnwidth, keepaspectratio]{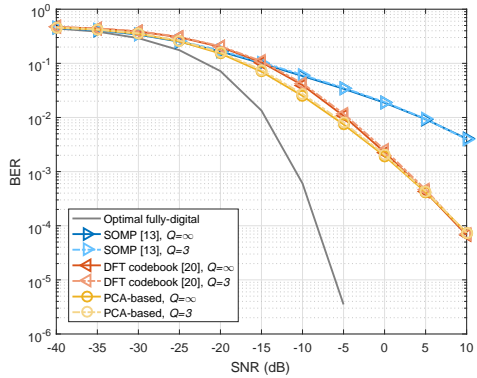}\\
	\caption{BER performance comparison of different hybrid precoding schemes, where both transmitter and receiver employ $8\times8$ UPA, $N_t^{\rm RF}=N_r^{\rm RF}=4$, and $K=128$.}\label{fig_BER}
\end{figure}
In Fig. \ref{fig_BER}, we evaluate the BER performance of the proposed PCA-based hybrid precoding scheme, SOMP-based hybrid precoding scheme, and DFT codebook-based hybrid precoding scheme, where both hybrid precoder and combiner are jointly investigated, 16 QAM is adopted for transmission, and $Q=\infty$ and $Q=3$ are considered, respectively. For simplicity, only $K=128$ subcarriers are considered in the system. As shown in Fig. \ref{fig_BER}, the proposed PCA-based hybrid precoding scheme outperforms the conventional SOMP-based and DFT codebook-based hybrid precoding schemes in BER performance. Meanwhile, the SOMP-based hybrid precoding scheme has the worst BER performance, especially at high SNR. When BER$=10^{-2}$ is considered, we can observe that the proposed scheme outperforms the DFT-codebook-based hybrid precoding scheme and the SOMP-based hybrid precoding scheme by approximately 2 dB and 11 dB, respectively.

\section{Conclusions}
This paper has proposed a PCA-based hybrid precoder
and combiner for wideband mmWave FD-MIMO backhaul/fronthaul channels. To design the precoder, we use the PCA algorithm to extract the principal components from the optimal fully-digital precoders of all subcarriers and choose its phase angles as the angles of RF precoder. Moreover, the RF combiner at the receiver can be built by implementing the weighted PCA, and baseband combiner can be designed by using the weighted LS. Simulations further verify both the better SE and BER performance of the proposed scheme than conventional schemes. The scheme proposed in this paper can also be used in cellular communications. However, smaller receive antenna array and multiple users are often considered in cellular communications, which will be investigated in future work.


\section*{Acknowledgment}
This work was supported by the National Natural Science Foundation of China (Grant Nos. 61471037, 61701027, and 61201181), the Beijing Natural Science Foundation (Grant No. 4182055), Huawei Innovation Research Program (HIRP), and Youth Project of China Academy of Information and Communications Technology.




%

\end{document}